\documentclass{iacrtrans}
\usepackage[utf8]{inputenc}

\author{Ge Yao\inst{1} \and Udaya Parampalli\inst{1}}
\institute{The University of Melbourne, Australia\\ \email{gyao1@student.unimelb.edu.au}} 
\title[\texttt{iacrtrans} class documentation]{Generalized NLFSR Transformation Algorithms and Cryptanalysis of the Class of Espresso-like Stream Ciphers}

\begin{document}
\maketitle


\begin{abstract}
  Lightweight stream ciphers are highly demanded in IoT applications. In order to optimize the hardware performance, a new class of stream cipher has been proposed. The basic idea is to employ a single Galois NLFSR with maximum period to construct the cipher. As a representative design of this kind of stream ciphers, Espresso is based on a 256-bit Galois NLFSR initialized by a 128-bit key. The $2^{256}-1$ maximum period is assured because the Galois NLFSR is transformed from a maximum length LFSR. However, we propose a Galois-to-Fibonacci transformation algorithm and successfully transform the Galois NLFSR into a Fibonacci LFSR with a nonlinear output function. The transformed cipher is broken by the standard algebraic attack and the R\o njom-Helleseth attack with complexity $\mathcal{O}(2^{68.44})$ and $\mathcal{O}(2^{66.86})$ respectively. The transformation algorithm is derived from a new Fibonacci-to-Galois transformation algorithm we propose in this paper. Compare to existing algorithms, proposed algorithms are more efficient and cover more general use cases. Moreover, the transformation result shows that the Galois NLFSR used in any Espresso-like stream ciphers can be easily transformed back into the original Fibonacci LFSR. Therefore, this kind of design should be avoided in the future.  
\end{abstract}

\section{Introduction}
\subsection{Background}
The Internet of Things (IoT) is a new technology paradigm envisioned as a global network of devices capable of interacting with each other \cite{lee2015internet}. The devices in IoT usually have limited computing resources and strict power constraints, hence reliable connections with highly energy efficient communication technology such as 5G are in need. In the various study of security solutions for 5G network, a new class of lightweight stream ciphers has been proposed \cite{DubrovaH17}. We refer to it as the Espresso-like ciphers since Espresso is the representative design. The basic idea of this new design is to employ a Nonlinear Feedback shift Register (NLFSR) in Galois configuration as the building block while most of the known stream ciphers are based on NLFSRs in Fibonacci configuration. By using a Galois NLFSR, the feedback functions of the cipher are much smaller compared to those based on Fibonacci NLFSRs, thereby leading to a lower propagation delay in implementation which is preferred in 5G wireless communication systems. The hardware performance analysis \cite{DubrovaH17} shows that Espresso is the fastest among the stream ciphers below 1500 GE.

Due to the fact that there is a lack of cryptanalysis being carried out on Galois NLFSR based stream ciphers, the security analysis of this kind of cipher is conducted on a transformed NLFSR F which resembles a Fibonacci NLFSR. The analysis result in~\cite{DubrovaH17} shows that the cipher is resistant to all known attacks. However, we find that the transformed NLFSR F is not equivalent to the original Galois NLFSR G in the cipher unless the output function of F is changed accordingly. There is a possibility that the security of the original NLFSR G based cipher is different from the so called equivalent NLFSR F based cipher. It is important to reevaluate the validity of this new design methodology for all kinds of potential weakness. 

\subsection{Related Work.} 
\subsubsection{Transformation Algorithms}
The Galois NLFSR G used in Espresso is transformed from a NLFSR F by applying the transformation algorithm proposed by Dubrova in 2009~\cite{Dubrova09}. This algorithm is the first successful attempt to transform Fibonacci NLFSR into Galois NLFSR. How to find the matching initial states for NLFSRs before and after transformation is presented in~\cite{Dubrova10}. Based on this algorithm, the author developed a method of constructing $n$-stage Galois NLFSRs with period $2^n - 1$ from $n$-stage maximum length LFSRs~\cite{Dubrova13}. This method laid the foundation for the design of the class of Espresso-like ciphers. The NLFSR F is constructed from a 256-stage maximum length LFSR. Then the transformation algorithm~\cite{Dubrova09} is used to convert F into the Galois NLFSR G. Therefore, the period of the output sequence of G is $2^{256} - 1$, which provides very good statistical properties for the cipher. Later, the author extend the transformation algorithm to handle Galois-to-Galois case in~\cite{Dubrova14}. This algorithm is the generalized version of the Fibonacci-to-Galois transformation algorithm~\cite{Dubrova09}. The equivalence between the two NLFSRs before and after transformation is hold if and only if they both are "uniform". The definition of "uniform" is described in Section 2. 

For the Galois-to-Fibonacci case, there are only a few results have been published. In 2013, Lin~\cite{Lin13} proposed a transformation from a Galois NLFSR to a Fibonacci NLFSR. This algorithm targets at Galois NLFSRs more general than the "uniform" case and studied the properties of the output sequences of all the bits in the Galois NLFSRs. Another algorithm~\cite{LuLHLC18} is proposed by using a mathematical tool named the semi-tensor product of matrices. However, the complexity of the proposed algorithm is $\mathcal{O}(2^n)$, where $n$ is the length of the NLFSR. This method is not applicable in stream ciphers since the size of the NLFSR is usually larger than 80-bit due to security concerns. Besides, the common problem in all the discussed algorithms is that the output function of an NLFSR is assumed to only tap from the 0th bit, which is infeasible in stream ciphers where the output takes multiple bits from the NLFSR. Furthermore, it is pointed out by the author in~\cite{Dubrova14} that the sequence of states of the two NLFSRs before and after transformation differ in several bit positions. How to efficiently and correctly transform more generalized NLFSRs with output function taken arbitrary taps from the NLFSR remain unsolved. 

\subsubsection{Security Attacks}

Shortly after the Espresso being published, a related key chosen-IV attack is proposed by Wang et.al~\cite{WangL17}. This attack is mounted on a variant of Espresso cipher denoted by Espresso-a. Similar as the transformation between NLFSR G and F described in Espresso, this variant is transformed from G by using the algorithm in~\cite{Dubrova09}. This attack recovers the 128-bit secret key with complexity $\mathcal{O}(2^{64})$. However, the output function of the variant Espresso-a is the same as that of Espresso. This causes the same issue of the security analysis in the Espresso~\cite{DubrovaH17} as we mentioned before. In another published paper \cite{ZWespresso}, an algebraic attack is mounted on Espresso cipher. Based on the fact that the Galois NLFSR used in Espresso cipher is constructed from a LFSR, the authors believe there must be a bijection between the NLFSR and the LFSR. They use an unknown mathematical software to find the filter function of the LFSR and present it only by parameters instead of the exact function, hence making it difficult to verify the result. Their method works only when the original LFSR is known to the attacker. 

In this paper, we propose a new and original transformation algorithm which is able to convert the whole class of Espresso-like cipher into LFSR filter generators and the proposed algorithm works even without knowing the original LFSR. Then we use algebraic attacks to break the transformed Espresso cipher. Algebraic attacks are very powerful to cryptanalyse LFSR-based stream ciphers. The main idea is to solve a system of algebraic equations between key bits and output bits. The core of the attack is to reduce the degree of the equations so that linearization method or XL method can be used to solve these equations efficiently. In 2003, Courtois et.al~\cite{CourtoisM03} proposed a general algebraic attack. They multiply the algebraic equations by well-chosen multivariate polynomials to lower the degree of these equations. Several subsequent researches have been published to speed up the solving process~\cite{Courtois03,HawkesR04,RonjomH07}. The attack in \cite{RonjomH07} is the most efficient attack among existing algebraic attacks.  Meanwhile, several approaches have been proposed to evaluate the ability of the output functions against algebraic attack which is denoted as the algebraic immunity~\cite{MeierPC04,DalaiGM04}. How to construct such high algebraic immunity Boolean functions is presented in~\cite{ArmknechtCGKMR06,CarletDGM06,CarletF08}.

\subsection{Our Contributions.}

The contribution of this paper is listed below.

1. We point out that the common problem in existing transformation algorithms is the output function is assumed to be $f_z = x_0$. However, in real NLFSR-based stream ciphers, the output function usually takes multiple taps from the NLFSR. In order to solve this problem, we develop an idea of compensating the output function and feedback functions of the NLFSR during transformation.

2. Based on the compensation idea, we propose a Fibonacci-to-Galois NLFSR transformation algorithm and a Galois-to-Fibonacci NLFSR transformation algorithm. Compare to existing transformation algorithms, the proposed algorithms are more efficient and generalized. First, the proposed algorithms has linear complexity and applicable on NLFSRs with arbitrary length. Second, bith algorithms cover more general cases compare to the "uniform" case in \cite{Dubrova09}. Third, no matter which taps from the NLFSR are taken in the output function, both algorithms show how to construct the corresponding output function and compute the initial value for the transformed NLFSR.

2. We customise the Galois-to-Fibonacci transformation algorithm to an Uniform\_Galois-to-Fibonacci algorithm to deal with the class of Espresso-like stream ciphers. The result shows that the Galois NLFSR used in any of this kind of ciphers can be transformed to a Fibonacci LFSR. The cipher after transformation is equivalent to a linear filter generator based stream cipher. 

3. We analyze the security of Espresso cipher on its transformed version and find out that it can be easily broken by typical attacks against linear filter generators. We apply the standard algebraic attack and the R\o njom-Helleseth attack and break the transformed cipher with complexity $\mathcal{O}(2^{68.44})$ and $\mathcal{O}(2^{66.86})$ respectively. We discuss other related attacks and conclude that this design method should not be used in the future. 

\subsection{Outline.}
In Section 2, we describe some necessary preliminaries and present the design specification of the Espresso cipher. Then we propose generalized NLFSR transformation algorithms in Section 3. In Section 4, we apply the proposed Uniform\_Galois-to-Fibonacci algorithm on the Galois NLFSR with period $2^n-1$ and show the result of transformation of G in Espresso cipher as an example. We then mount algebraic attacks on the transformed cipher and discuss the overall security of the class of Espresso-like stream ciphers in Section 5. We conclude in Section 6.

\section{Preliminaries}
\subsection{FSRs}

An $n$-bit Feedback Shift Register (FSR) consists of $n$ binary storage elements. We refer to each storage element as a \textbf{\textit{stage}} or \textbf{\textit{tap}} represented by $x_i, i \in [0, n-1]$. A FSR is controlled by a system clock, the bit values stored in it are shifted one stage to the left and the last stage is updated by a feedback function $f(X) = f(x_0, \ldots, x_{n-1})$ which takes any taps as input. A Linear Feedback Shift Register (LFSR) is a FSR with a linear feedback function. A Nonlinear Feedback Shift Register (NFSR) is a FSR with a nonlinear feedback function. The output function and feedback functions in a FSR are Boolean functions.

We denote the addition and multiplication in $GF(2)$ as "$\oplus$" and "$\cdot$" respectively throughout the paper. 

\begin{definition} \label{def:ANF}
The algebraic normal form (ANF) of a Boolean function $f:\{0, 1\}^n \rightarrow \{0, 1\}$ is a polynomial in GF(2) of type \cite{Dubrova09} 
\begin{center}
    $f_{n-1} = \sum_{i=0}^{2^n-1}c_ix_0^{i_0}\ldots x_{n-1}^{i_{n-1}}$,
\end{center}
where $c_i, i_0, \ldots, i_{n-1} \in \{0, 1\}$. Each term in a Boolean function is called a \textbf{\textit{monomial}}. For example, in function $f = x_0 \oplus x_1x_2$, both $x_0$ and $x_1x_2$ are called a monomial.
\end{definition}

\textbf{\textit{dep($\cdot$).}} A dependence list of a Boolean function $g$ denoted by $dep(g)$ is the list of the indexes of all the involved stages of FSR. For example, $g = x_0 \oplus x_1 \oplus x_2x_3$, then $dep(g) = [0, 1, 2, 3]$. If the indexes in the dependence list of $g$ are increased by $i$, we denote it as $g|_{+i}$. For example, $g = x_2x_3$, then $dep(g) = [2, 3]$, $g|_{+1} = x_3x_4$ and $dep(g|_{+1}) = [3, 4]$. If the function $g$ only has one term, then it is a monomial denoted as $m$. $dep(m)$ is the dependence list of the monomial.

We denote the the set of all the bits in FSR as an \textbf{\textit{internal state}} $X^t = \{x_0^t, \ldots, x_{n-1}^t\}$ at each clock $t$. The \textbf{\textit{initial value}} $X^0 = \{x_0^0, \ldots, x_{n-1}^0\}$ is the first internal state of the FSR. In a FSR, the output sequence is default of the sequence generated by tap $x_0$. In FSR-based stream ciphers, the output sequence is usually generated by a output function $f_z(x_0, \ldots, x_{n-1})$ which takes any tap from the FSR.

\textbf{\textit{Configurations.}} A FSR can be implemented in two kinds of configurations, namely Fibonacci and Galois configuration. In the \textbf{\textit{Fibonacci configuration}}, the feedback is only applied to the last stage. In the \textbf{\textit{Galois configuration}}, the feedback can be applied to every stage. An example of 4-bit Fibonacci NLFSR is given in Figure 1. The feedback $f_3$ is only applied to the last stage $x_3$. In Figure 2, we present a 4-bit Galois NLFSR. The feedback $f_2$ is applied to $x_2$ and the feedback $f_3$ is applied to $x_3$.

\begin{figure}[h]
\includegraphics[width=0.9\linewidth, height=3.5cm]{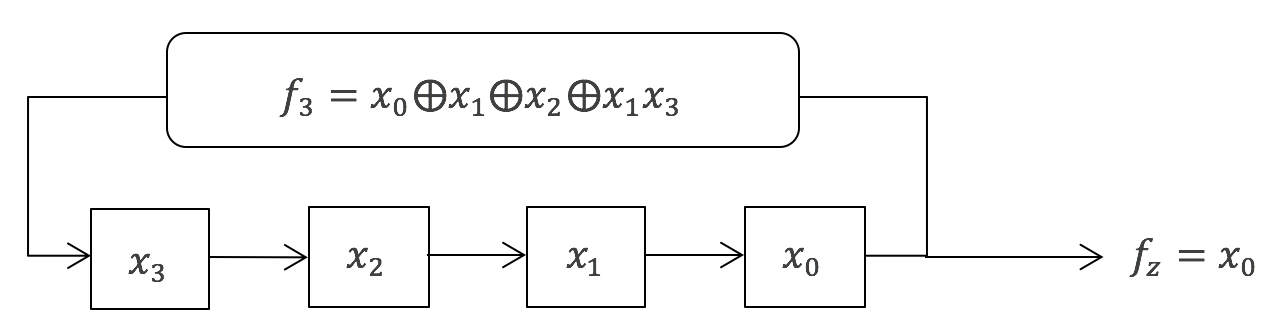}
\caption{An example of a 4-bit Fibonacci NLFSR.}
\end{figure}
\begin{figure}[h]
\includegraphics[width=0.9\linewidth, height=4cm]{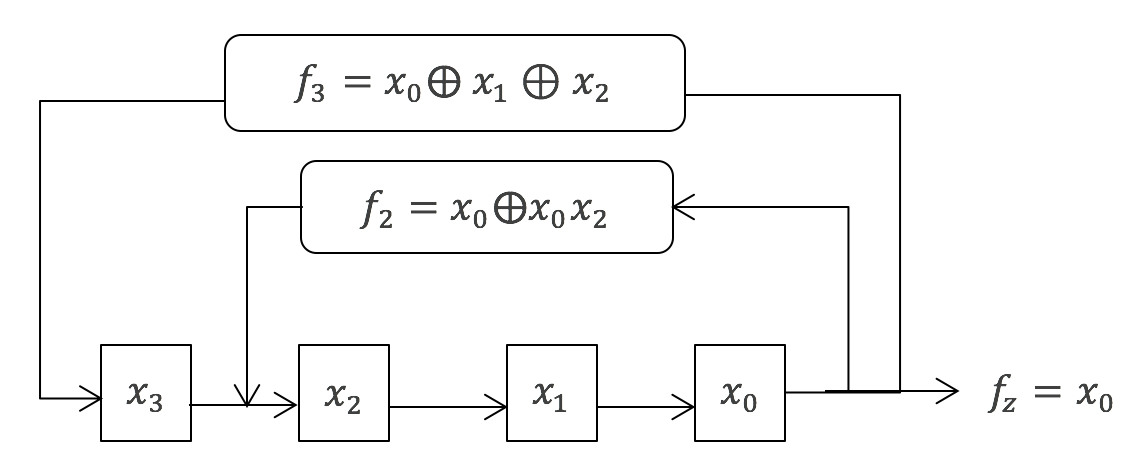}
\caption{An example of a 4-bit Galois NLFSR.}
\end{figure}

\begin{definition} \label{def:uniform}
Given an $n$-bit NLFSR. Suppose the smallest index of the stage that updated by a feedback is $\tau$. The NLFSR is \textbf{\textit{uniform}} if the feedback functions are of type

\begin{equation} \label{eq:uniform}
    \begin{aligned}
    &f_{n-1} = x_0 \oplus g_{n-1}(x_1, \ldots, x_\tau)\\
    &\ldots\\
    &f_\tau = x_{\tau+1} \oplus g_\tau(x_0, \ldots, x_\tau)\\
    &f_i = x_{i+1} \quad for \quad i \in [0, \tau-1],
    \end{aligned}
\end{equation}
where $g_i, i \in [\tau, n-1]$ is a Boolean function with $max(dep(g_i)) \leq \tau$ and $0 \notin dep(g_{n-1})$. 
\end{definition}

\subsection{Galois NLFSRs with Period $2^n -1$}
The problem of constructing NLFSR generating sequence with given long period is an open problem in studying the theory of NLFSR. In order to solve this problem, a scalable method is proposed to construct Galois NLFSRs with period $2^n -1$ from a $n$-stage maximum-length LFSR~\cite{Dubrova13}. The main idea is to employ the transformation algorithm presented in~\cite{Dubrova09} to shift nonlinear monomials from the feedback function of the last stage to lower stages in the LFSR. Since the FSR before and after transformation are equivalent, the transformed FSR is a Galois NLFSR with period $2^n -1$. The validity of this method is proved by the following theorem  proposed in~\cite{Dubrova13}.

\begin{theorem}
Let $N$ be an $n$-stage NLFSR with the feedback functions of type 
\begin{align*}
    &f_{n-1} = f_L \oplus f_N(x_1, x_2, \ldots, x_{n-2})\\
    &f_{n-2} = x_{n-1} \oplus f_N(x_0, x_1, \ldots, x_{n-3})\\
    &f_{n-3} = x_{n-2}\\
    &\qquad\ldots\\
    &f_0 = x_1,
\end{align*}
where $f_L$ is a linear Boolean function of type
\begin{center}
    $f_L = x_0 \oplus c_1x_1 \oplus \ldots \oplus c_{n-2}x_{n-2}$,
\end{center}
where $c_i \in \{0, 1\}, i \in [1, n-2]$. If the corresponding character polynomial 
\begin{center}
    $g(x) = 1 + c_1x + c_2x^2 + \ldots + c_{n-2}x^{n-2} + x^n$
\end{center}
is primitive. Then N has period $2^n -1$.
\end{theorem}

The Galois NLFSR defined in Theorem 1 is one of the results of transformation of the maximum LFSR with feedback function $f_{n-1}= x_0 \oplus f_L(x_1, \ldots, x_{\tau}) \oplus f_N(x_1, x_2, \ldots, x_{n-\tau}) \oplus f_N(x_1, x_2, \ldots, x_{n-\tau})$. According to the transformation algorithm in~\cite{Dubrova09}, the monomials in $f_N$ can be further shifted to lower stages as long as the NLFSR after shifting is "uniform". Formally, we define the transformed NLFSR by feedback functions  
\begin{equation*}
    \begin{aligned}
    &f_{n-1} = x_0 \oplus g_{n-1}(x_1, \ldots, x_{n-\tau})\\
    &f_{n-2} = x_{n-1} \oplus g_{n-2}(x_0, \ldots, x_{n -\tau - 1})\\
    &\qquad\ldots\\
    &f_{\tau} = x_{\tau + 1} \oplus g_{\tau}(x_0, x_1)\\
    &f_{\tau-1} = x_{\tau}\\
    &\qquad\ldots\\
    &f_0 = x_1.
    \end{aligned}
\end{equation*}

After shifting the monomials in $f_N$ to different stages, the depth of circuits implementing the feedback functions is reduced, resulting in faster and small NLFSRs. This design method seems to be very promising for the 5G applications. 

\subsection{Espresso Cipher}
The main building block of the Espresso cipher~\cite{DubrovaH17} is a 256-bit Galois NLFSR G, which is constructed by the method present in Section 2.2. The designers first choose a maximum-length Fibonacci LFSR with feedback function
\begin{center}
    $f_{255} = x_0 \oplus x_{12} \oplus x_{48} \oplus x_{115} \oplus x_{133} \oplus x_{213}$.
\end{center}

Then they construct a Galois NLFSR F from the LFSR as
\begin{align*}
    f_{255} &= x_0 \oplus x_{12} \oplus x_{48} \oplus x_{115} \oplus x_{133} \oplus x_{213} \oplus x_{41}x_{70} \oplus x_{46}x_{87}\\
    f_{217} &= x_{218} \oplus x_3x_{32} \oplus x_8x_{49} \oplus x_{14}x_{72} \oplus x_{17}x_{92} \oplus x_{24}x_{119} \oplus x_{36}x_{145} \oplus x_{49}x_{72}x_{92}x_{119},
\end{align*}
and all remaining feedback functions are of type $f_i = x_{i+1}$.

The monomials in the functions are shifted to lower stages. The Galois NLFSR F is transformed to G
\begin{equation*}
\begin{aligned}
f_{255} &= x_0 \oplus x_{41}x_{70}\\
f_{251} &= x_{252} \oplus x_{42}x_{83} \oplus x_8\\
f_{247} &= x_{248} \oplus x_{44}x_{102} \oplus x_{40}\\
f_{243} &= x_{244} \oplus x_{43}x_{118} \oplus x_{103}\\
f_{239} &= x_{240} \oplus x_{46}x_{141} \oplus x_{117}\\
g_{235} &= x_{236} \oplus x_{67}x_{90}x_{110}x_{137}\\
f_{231} &= x_{232} \oplus x_{50}x_{159} \oplus x_{189}\\
f_{217} &= x_{218} \oplus x_3x_{32}\\
f_{213} &= x_{214} \oplus x_4x_{45}\\
f_{209} &= x_{210} \oplus x_6x_{64}\\
f_{205} &= x_{206} \oplus x_5x_{80}\\
f_{201} &= x_{202} \oplus x_8x_{103}\\
f_{197} &= x_{198} \oplus x_{29}x_{52}x_{72}x_{99}\\
f_{193} &= x_{194} \oplus x_{12}x_{121}.
\end{aligned}
\end{equation*}
The rest of the functions remain as $f_i = x_{i + 1}$.

The output function for Espresso is a nonlinear function with 20 variables
\begin{center}
$f_z = x_{80} \oplus x_{99} \oplus x_{137} \oplus x_{227} \oplus x_{222} \oplus x_{187} \oplus x_{243}x_{217} \oplus x_{247}x_{231} \oplus x_{213}x_{235} \oplus x_{255}x_{251} \oplus x_{181}x_{239} \oplus x_{174}x_{44} \oplus x_{164}x_{29} \oplus x_{255}x_{247}x_{243}x_{213}x_{181}x_{174}$.
\end{center}

The cipher is initialized by a 128-bit key $k_i, 0 \leq i \leq 127$ and an 96-bit initialization value IV. The initial internal states are 
\begin{align*}
x_i &= k_i, 0 \leq i \leq 127\\
x_i &= IV_{i -128}, 128 \leq i \leq 223\\
x_i &= 1, 224 \leq i \leq 254\\
x_i &= 0, i = 255.
\end{align*}

At the initialization phase, the cipher would be clocked 256 times, the output bit is xored with the stages $x_{255}$ and $x_{217}$
\begin{align*}
    f_{255} &= x_0 \oplus x_{41}x_{70} \oplus f_z\\
    f_{217} &= x_{218} \oplus x_3x_{32} \oplus f_z.
\end{align*}

\section{Transformation Algorithms}
In this section, we analyze the difference in internal states of the two NLFSRs before and after transformation. In order to fix the difference, we develop an idea of compensating the feedback function. Based on this idea, we propose several transformation algorithms for different application cases respectively. In the proposed algorithms, we consider the output function taken arbitrary taps from the NLFSR, and show the method to construct the corresponding output function and initial value for the transformed NLFSR. 

\subsection{Difference in Internal States}
The transformation of a NLFSR is done by shifting the monomials from a feedback function to another one. We define shifting a monomial $m$ from feedback function $f_a$ to $f_b$ as below \cite{Dubrova09}.

\begin{definition} \label{def:shifting}
Let $f_a$ and $f_b$ be feedback functions of bit $x_a$ and $x_b$ of an $n$-bit NLFSR, where $a, b \in [0, n-1]$ and $a \neq b$. The operation \textit{shifting} moves a monomial $m$ from $f_a$ to $f_b$, denoted by $f_a \xrightarrow{m} f_b$. The index of each variable $x_i$ in $m$ is changed to $x_{(i-a+b)\ \mathrm{mod}\ n}$. The shifted monomial in $f_b$ is denoted by $m|_{(-a+b)\ \mathrm{mod}\ n}$.
\end{definition}

Suppose we shift a monomial $m$ with $dep(m) \subseteq [1, n-1]$ in the feedback function $f_{n-1}$ of a Fibonacci NLFSR  to another feedback function $f_b$ following Definition \ref{def:shifting}. The Fibonacci NLFSR is transformed to the Galois NLFSR
\begin{equation*}
    \begin{aligned}
    &f_{n-1} = x_0 \oplus g_{n-1}(x_1, \ldots, x_{n-1}) \oplus m\\
    &\qquad\ldots\\
    &f_b = x_{b + 1} \oplus m|_{-(n-1-b)}\\
    &f_{b-1} = x_b\\
    &\qquad\ldots\\
    &f_0 = x_1.
    \end{aligned}
\end{equation*}

In the feedback function $f_{n-1}$, $m$ is xored to it representing that $m$ is moved from this function and $m|_{-(n-1-b)}$ is xored to $f_b$ showing that the monomial is shifted to $f_b$.

Now we analyze the internal states of NLFSRs before and after shifting a monomial. Suppose the internal states at clock $t$ of the Fibonacci NLFSR and the Galois NLFSR are denoted by $X^t = \{x_0^t, \ldots, x_{n-1}^t\}$ and $\hat{X}^t = \{\hat{x}_0^t, \ldots, \hat{x}_{n-1}^t\}$ respectively and $X^t = \hat{X}^t$. Then the internal states at next clock of the Fibonacci NLFSR are computed as
\begin{equation*}
    \begin{aligned}
    &x_{n-1}^{t+1} = x_0^t \oplus g_{n-1}(x_1^t, \ldots, x_{n-1}^t)\\
    &\qquad\ldots\\
    &x_b^{t+1} = x_{b + 1}^t\\
    &x_{b-1}^{t+1} = x_b^t\\
    &\qquad\ldots\\
    &x_0^{t+1} = x_1^t
    \end{aligned}
\end{equation*}
and the internal states at next clock of the Galois NLFSR are computed as
\begin{equation*}
    \begin{aligned}
    &\hat{x}_{n-1}^{t+1} = \hat{x}_0^t \oplus g_{n-1}(\hat{x}_1^t, \ldots, \hat{x}_{n-1}^t) \oplus \hat{m}^t\\
    &\qquad\ldots\\
    &\hat{x}_b^{t+1} = \hat{x}_{b + 1}^t \oplus \hat{m}|_{-(n-1-b)}^t\\
    &\hat{x}_{b-1}^{t+1} = \hat{x}_b^t\\
    &\qquad\ldots\\
    &\hat{x}_0^{t+1} = \hat{x}_1^t.
    \end{aligned}
\end{equation*}

We compare $X^{t+1}$ and $\hat{X}^{t+1}$ and find that they differ at bit $x_{n-1}$ and $x_b$. If we continue running the two NLFSRs, the internal states will differ in every bit. In \cite{Dubrova09}, it is proved that if the Galois NLFSR after transformation is "uniform", then the internal states only differ in bits $x_{b+1}, \ldots, x_{n-1}$ between the two NLFSRs if the initial value of the Galois NLFSR is calculated by the theorem in \cite{Dubrova10}. Therefore, if the output function of the NLFSR takes inputs from any bit among $x_{b+1}, \ldots, x_{n-1}$, then the two NLFSRs will generate different output sequences. In order to solve this problem. We develop the idea of compensating feedback functions of the NLFSR during the transformation. 

\begin{definition} \label{def:list}
Given an $n$-bit NLFSR, suppose we shift a monomial $m$ from $f_a$ to $f_b$ where $a, b \in [0, n-1]$ and $a > b$, the compensation list is constructed as $C = [0, \ldots, 0, m|_{-(a-b)}, \ldots, m|_{-1}, 0, \ldots, 0]$ where $C[i] = 0$ for $i \in [0, b]$ or $[a+1, n-1]$ and $C[i] = m|_{-(a-i+1)}$ for $i \in [b+1, a]$.
\end{definition}

\begin{definition} \label{def:compensating}
Given a Boolean function $f(x_0, \ldots, x_{n-1})$ and a compensation list $C$.  The operation \textit{compensating} replaces the tap $x_i$ in $f$ by $x_i \oplus C[i]$ iteratively from $i = n-1$ to $i = 0$.
\end{definition}

Based on the compensation idea, we prove that the differences between internal states of the two NLFSRs before and after transformation satisfy a relationship in the following Lemma. 

\begin{lemma} \label{lem1}
Given an $n$-bit Fibonacci NLFSR, we shift a monomial $m$ with $dep(m) \subseteq [1, n-1]$ from feedback function $f_{n-1}$ to $f_b$ by Definition \ref{def:shifting}. We construct the compensation list by Definition \ref{def:list} and compensate all the $g_j, j \in [0, n-1]$ in feedback functions of the transformed NLFSR according to Definition \ref{def:compensating} and calculate the initial value for it by the following equations
\begin{equation}\label{eq:r}
    \begin{aligned}
    &\hat{x}_i^t = x_i^t \oplus m|_{-(n-i)}^t \quad for \quad i \in [b+1, n-1];\\
    &\hat{x}_i^t = x_i^t \quad for \quad i \in [0, b].
    \end{aligned}
\end{equation}
If the index of $f_b$ satisfy $b \in [n-1-min(dep(m)), n-2]$, then internal states at any clock $t$ of the two NLFSRs before and after transformation satisfy \eqref{eq:r}.
\end{lemma}

\begin{proof}
We prove this Lemma by induction. First we suppose the internal states of the two NLFSRs before and after transformation satisfy the relation \eqref{eq:r} for clock $t = k$, then we prove the relation holds for the next clock $t = k + 1$.

According to feedback functions in Definition 2, internal states of the Fibonacci NLFSR at clock $t = k + 1$ are

\begin{equation}\label{eq:sf}
    \begin{aligned}
    &x_{n-1}^{k+1} = x_0^k \oplus g_{n-1}(x_1^k, \ldots, x_{n-1}^k)\\
    &x_{n-2}^{k+1} = x_{n-1}^k\\
    &\qquad\ldots\\
    &x_b^{k+1} = x_{b + 1}^k\\
    &x_{b-1}^{k+1} = x_b^k\\
    &\qquad\ldots\\
    &x_0^{k+1} = x_1^k.
    \end{aligned}
\end{equation}

Since the compensation of the feedback functions of the transformed NLFSR is done by following Definition \ref{def:compensating}, the feedback functions of the transformed NLFSR are

\begin{equation*}
    \begin{aligned}
    &f_{n-1} = x_0 \oplus g'_{n-1}(x_0, \ldots, x_{n-1})\\
    &f_{n-2} = x_{n-1}\\
    &\qquad\ldots\\
    &f_{b+1} = x_{b+2}\\
    &f_b = x_{b + 1} \oplus m|_{-(n-1-b)}\\
    &f_{b-1} = x_b\\
    &\qquad\ldots\\
    &f_0 = x_1,
    \end{aligned}
\end{equation*}
where $g'_{n-1}$ is the result of compensating $g_{n-1} \oplus m$ by the compensation list $C = [0, \ldots, 0, m|_{-(n-1-b)}, \ldots, m|_{-1}]$ iteratively. The shifted monomial $m|_{-(n-1-b)}$ in $f_b$ is not compensated because $max(dep(m)) \leq n-1$ is satisfied. Hence, $max(dep(m|_{-(n-1-b)})) \leq b$. The internal states of this NLFSR at clock $t = k + 1$ are calculated as

\begin{equation}\label{eq:sg}
    \begin{aligned}
    &\hat{x}_{n-1}^{k+1} = \hat{x}_0^k \oplus g'_{n-1}(\hat{x}_0^k, \ldots, \hat{x}_{n-1}^k)\\
    &\hat{x}_{n-2}^{k+1} = \hat{x}_{n-1}^k\\
    &\qquad\ldots\\
    &\hat{x}_{b+1}^{k+1} = \hat{x}_{b+2}^k\\
    &\hat{x}_b^{k+1} = \hat{x}_{b + 1}^k \oplus \hat{m}|_{-(n-1-b)}^k\\
    &\hat{x}_{b-1}^{k+1} = \hat{x}_b^k\\
    &\qquad\ldots\\
    &\hat{x}_0^{k+1} = \hat{x}_1^k.
    \end{aligned}
\end{equation}

Since the relation \eqref{eq:r} holds for $t = k$, we have $\hat{x}_i^k = x_i^k$ for $i \in [0, b]$. Hence, $\hat{x}_i^{k+1} = x_i^{k+1}$ for $i \in [0, b-1]$ is valid, which means the relation \eqref{eq:r} holds for internal states of bits $i \in [0, b-1]$ at clock $t = k + 1$. 

As for bit $x_b$, since $b \in [n-1-min(dep(m)), n-2]$, we have $dep(m) \subseteq [n-1-b, n-1]$. Therefore, the indexes of the compensation elements satisfy
\begin{equation}\label{eq:index}
    dep(m|_{-(n-i)}) \subseteq[i-b-1, i-1], \quad for \quad i \in [b+1, n-1].
\end{equation}
From \eqref{eq:sg} and $\hat{x}_{b+1}^k = x_{b+1}^k \oplus m|_{-(n-(b+1))}^k$, we calculate the internal state of bit $x_b$ as $\hat{x}_b^{k+1} = \hat{x}_{b + 1}^k \oplus \hat{m}|_{-(n-1-b)}^k = x_{b+1}^k \oplus m|_{-(n-(b+1))}^k \oplus \hat{m}|_{-(n-1-b)}^k$. Then from \eqref{eq:index}, we have $dep(m|_{-(n-1-b)}) \subseteq [0, b]$, hence, $m|_{-(n-(b+1))}^k = \hat{m}|_{-(n-1-b)}^k$. Therefore, we get $\hat{x}_b^{k+1} = x_{b+1}^k = x_b^{k+1}$, which means the relation \eqref{eq:r} holds for internal state of bit $x_b$ at clock $t = k + 1$. 

As for each bit $x_i, i \in [b+1, n-2]$, from \eqref{eq:sg} and the assumption that the relation \eqref{eq:r} holds for $t = k$, we have $\hat{x}_{i}^{k+1} = \hat{x}_{i+1}^k = x_{i+1}^k \oplus m|_{-(n-(i+1))}^k$. From(4), we have $dep(m|_{-(n-(i+1))}) \subseteq [i-b, i]$. Then from \eqref{eq:sf}, we get $x_i^{k+1} = x_{i+1}^k$ and $m|_{-(n-(i+1))}^k = m|_{-(n-i)}^{k+1}$. Therefore, we get $\hat{x}_{i}^{k+1} = x_{i}^{k+1} \oplus m|_{-(n-i)}^{k+1}$, which means the relation \eqref{eq:r} holds for internal states of bits $i \in [b+1, n-2]$ at clock $t = k + 1$.

As for the last bit $x_{n-1}$, from \eqref{eq:sg} and \eqref{eq:sf}, we calculate
\begin{equation}\label{eq:iter}
   \begin{aligned} 
    \hat{x}_{n-1}^{k+1} &= \hat{x}_0^k \oplus g'_{n-1}(\hat{x}_0^k, \ldots, \hat{x}_{n-1}^k)\\
    &= x_0^k \oplus g'_{n-1}(\hat{x}_0^k, \ldots, \hat{x}_{n-1}^k)\\
    &= x_{n-1}^{k+1} \oplus g_{n-1}(x_1^k, \ldots, x_{n-1}^k) \oplus g'_{n-1}(\hat{x}_0^k, \ldots, \hat{x}_{n-1}^k).
    \end{aligned}
\end{equation}

Now we have to prove that $g_{n-1}(x_1^k, \ldots, x_{n-1}^k) \oplus g'_{n-1}(\hat{x}_0^k, \ldots, \hat{x}_{n-1}^k) = m|_{-1}^{k+1}$, we distinguish two cases:

Case 1: $max(dep(g_{n-1})) \leq b$. In this case, we have $max(dep(g_{n-1}\oplus m)) \leq b$, then no replacing of the bits in $g'_{n-1}$ takes place during compensation. Hence, we have $g'_{n-1} = g_{n-1} \oplus m$. then we get $g_{n-1}(x_1^k, \ldots, x_{n-1}^k) \oplus g'_{n-1}(\hat{x}_0^k, \ldots, \hat{x}_{n-1}^k) = m^k = m|_{-1}^{k+1}$. Therefore, we have proved that $\hat{x}_{n-1}^{k+1} = x_{n-1}^{k+1} \oplus m|_{-1}^{k+1}$, which means the relation \eqref{eq:r} holds for internal states of bit $n-1$ at clock $t = k + 1$.

Case 2: $max(dep(g_{n-1})) > b$. In this case, $g'_{n-1}$ is the iterative compensation result of $g_{n-1} \oplus m$. During the compensation, the tap $x_i$ in every $g_j, 0 \leq j \leq n-1$ is replaced by $x_i \oplus C[i]$ iteratively from $i = n-1$ to $i = b+1$, where $C[i] = m|_{-(n-i)}$. Each iterative compensation step for $x_i$ makes sure that not only the taps in $g_{n-1} \oplus m$ but also the taps in previous compensated values $C[n-1], ..., C[i+1]$ are replaced. Consequently, the tap $x_{b+1}$ appeared in $g'_{n-1}$ is replaced by $x_{b+1} \oplus m|_{-(n-(b+1))}$. The tap $x_{b+2}$ appeared in $g'_{n-1}$ is replaced by $x_{b+2} \oplus m'|_{-(n-(b+2))}$, where $m'|_{-(n-(b+2))}$ is $m|_{-(n-(b+2))}$ with $x_{b+1}$ replaced by $x_{b+1} \oplus m|_{-(n-(b+1))}$. The rest of taps are replaced similarly. Then from the assumption and \eqref{eq:index}, we calculate the internal states
\begin{equation} \label{eq:is}
    \begin{aligned}
    \hat{x}_{b+1}^k \oplus \hat{m}|_{-(n-(b+1))}^k &= x_{b+1}^k \oplus m|_{-(n-(b+1))}^k \oplus m|_{-(n-(b+1))}^k \\
    &= x_{b+1}^k\\
    \hat{x}_{b+2}^k \oplus \hat{m}'|_{-(n-(b+2))}^k &= \hat{x}_{b+2}^k \oplus m|_{-(n-(b+2))}^k (\hat{x}_1^k, \ldots, \hat{x}_{b+1}^k \oplus \hat{m}|_{-(n-(b+1))}^k)\\
    &= \hat{x}_{b+2}^k \oplus m|_{-(n-(b+2))}^k (x_1^k, \ldots, x_{b+1}^k)\\
    &= x_{b+2}^k,
    \end{aligned}
\end{equation}
the internal states of rest bits $\hat{x}_i^k \oplus \hat{m}'|_{-(n-i)}^k = x_i^k, i \in [b+3, n-1]$ are calculated similarly. Since the $g'_{n-1}$ is the iterative compensation result of $g_{n-1} \oplus m$, combined with the internal states we just calculated, we get
\begin{equation*}
    \begin{aligned}
    g'_{n-1}(\hat{x}_0^k, \ldots, \hat{x}_{n-1}^k) &= g_{n-1} \oplus m(\hat{x}_0^k, \ldots, \hat{x}_b^k, \hat{x}_{b+1}^k \oplus \hat{m}|_{-(n-(b+1))}^k,\\
    &\quad\quad\quad\quad \hat{x}_{b+2}^k \oplus \hat{m}'|_{-(n-(b+2))}^k \ldots, \hat{x}_{n-1}^k \oplus \hat{m}'|_{-(n-(n-1))}^k)\\
    &= g_{n-1} \oplus m(x_0^k, \ldots, x_{n-1}^k).
\end{aligned}
\end{equation*}

Therefore, we prove that \eqref{eq:iter} can be further calculated to
\begin{equation*}
    \begin{aligned}
    \hat{x}_{n-1}^{k+1} &= x_{n-1}^{k+1} \oplus g_{n-1}(x_1^k, \ldots, x_{n-1}^k) \oplus g'_{n-1}(\hat{x}_0^k, \ldots, \hat{x}_{n-1}^k)\\
    &= x_{n-1}^{k+1} \oplus m(x_0^k, \ldots, x_{n-1}^k)\\
    &= x_{n-1}^{k+1} \oplus m|_{-1}^{k+1}.
    \end{aligned}
\end{equation*}

Now we have proved the relation \eqref{eq:r} is valid for internal states at clock $t = k+1$ when the internal states at clock $t = k$ satisfy \eqref{eq:r}. Since the initial value also satisfy this relationship, we conclude that the internal states of the two NLFSRs before and after transformation have relationship \eqref{eq:r} for every clock $t$.
\end{proof}

\subsection{Proposed Transformation Algorithms}
Lemma \ref{lem1} shows that the differences between internal states of the two NLFSRs are fixed if we compensate the feedback functions after shifting a single monomial from $f_{n-1}$ to $f_b$. The only condition is that $n-1-min(dep(m)) \leq b \leq n-2$. Similarly, when we shift multiple monomials from $f_{n-1}$ to different feedback functions, a similar relationship will hold if the condition is satisfied for each monomial. Based on this idea, we propose a new Fibonacci-to-Galois transformation algorithm in Theorem \ref{Fib}. By using this algorithm, we are able to transform a Fibonacci NLFSR into more generalized Galois NLFSRs compare to the algorithm in \cite{Dubrova09}, in which the transformed Galois NLFSR must be "uniform". By reversing the compensation process, the Galois NLFSR can be transformed back to the Fibonacci NLFSR. A reverse algorithm denoted as Galois-to-Fibonacci transformation algorithm is proposed in Theorem \ref{Gal}. To be noted that both of the proposed algorithms cover the "uniform" case.

\begin{theorem}\label{Fib}
Given an $n$-bit Fibonacci NLFSR with an output function $f_z(x_0, \ldots, x_{n-1})$ and an initial value $X^0 = \{x_0^0, \ldots, x_{n-1}^0\}$, we shift $r$ monomials $m_0, m_1, \ldots, m_{r-1}$ from $f_{n-1}$ to feedback functions $f_{b_0}, \ldots, f_{b_{r-1}}, 0 \leq b_0 < b_1 < \ldots < b_{r-1} \leq n-2$ respectively. For each monomial, we construct a compensation list by Definition \ref{def:list} and xor all the lists together as a combined compensation list $C$. Then we use $C$ to compensate the feedback functions and the output function by Definition \ref{def:compensating}, and calculate the initial value as $\hat{X}^0 = \{\hat{x}_i^0 = x_i^0 \oplus C[i]^0, i\in [0, n-1]\}$. If indexes of feedback functions $f_{b_0}, \ldots, f_{b_{r-1}}$ satisfy that $b_j \in [n-1-min(dep(m_j)), n-2], j \in [0, r-1]$, then the transformed Galois NLFSR has feedback functions
\begin{equation} \label{eq:ff}
    \begin{aligned}
    &f_{n-1} = x_0 \oplus g'_{n-1}(x_1, \ldots, x_{n-1})\\
    &f_i = x_{i+1} \oplus g'_i(x_0, \ldots, x_i) \quad for \quad i \in [0, n-2] 
    \end{aligned}
\end{equation}
and it generates the same sequence as the Fibonacci NLFSR outputs.
\end{theorem}

\begin{proof}
According to Definition \ref{def:list}, the compensation list for each shifted monomial is
\begin{equation*}
    \begin{aligned}
    &C_0 = [0, \ldots, 0, m_0|_{-(n-1-b_0)}, \ldots, m_0|_{-1}]\\
    &C_1 = [0, \ldots, 0, m_1|_{-(n-1-b_1)}, \ldots, m_1|_{-1}]\\
    &\ldots\\
    &C_{r-1} = [0, \ldots, 0, m_{r-1}|_{-(n-1-b_{r-1})}, \ldots, m_{r-1}|_{-1}].\\
    \end{aligned}
\end{equation*}
Then the combined compensation list is $C = C_0 \oplus C_1 \oplus \ldots \oplus C_{r-1}$. It is easy to get $C[b_j+1] = C[b_j] \oplus m|_{-(n-1-b_j)}$ for $j \in [0, r]$.

Since the indexes of feedback functions to which monomials shifted satisfy $b_j \in [n-1-min(dep(m_j)), n-2], j \in [0, r-1]$, the indexes of the shifted monomial satisfy that $dep(m_j|_{-(n-1-b_j)}) \subseteq [0, b_j]$. Besides, the compensation value $C[b_j]$ has indexes smaller than $b_j$. Hence, after shifting all the monomials and compensating the feedback functions, the indexes in $g_{b_j}$ are in $[0, b_j]$, resulting in feedback functions \eqref{eq:ff}. 

In order to prove the two NLFSRs generate same sequence, we first prove that the differences between the internal states of the two NLFSRs are fixed as 
\begin{equation}\label{eq:r2}
    \begin{aligned}
    \hat{x}_i^t &= x_i^t \oplus C[i]^t \quad for \quad i \in [b_0+1, n-1]\\
    \hat{x}_i^t &= x_i^t \quad\quad\quad\quad for \quad i \in [0, b_0].
    \end{aligned}
\end{equation}

The proof for this relationship is simialr to the proof of Lemma \ref{lem1}. The only difference is the compensation list we use is a combination of all the lists for every monomial, but this does not affect that the resulted feedback functions \eqref{eq:ff} are the iterative compensation result of original feedback functions with monomials shifted. Specifically, $g'_{n-1}$ is $g_{n-1} \oplus m_0 \oplus \ldots \oplus m_{r-1}$ compensated by $C$ iteratively. $g'_{b_j}, j \in [0, r]$ is $m_j|_{-(n-1-b_j)}$ compensated by $C$ iteratively, and $g'_i = 0$ if no monomial shifted to $f_i$. Each iterative compensation step makes sure that not only the taps in $g_{n-1} \oplus m_0 \oplus \ldots \oplus m_{r-1}, m_j|_{-(n-1-b_j)}$ but also the taps in previous compensated values $C[n-1], ..., C[i+1]$ are replaced accordingly. Consequently, the tap $x_{b_0+1}$ is replaced by $x_{b_0+1} \oplus C[b_0+1]|_{-(n-(b_0+1))}$. The tap $x_{b_0+2}$ is replaced by $x_{b_0+2} \oplus C[b_0+2]'|_{-(n-(b_0+2))}$, where $C[b_0+2]'|_{-(n-(b_0+2))}$ is $C[b_0+2]|_{-(n-(b_0+2))}$ with $x_{b_0+1}$ replaced by $x_{b_0+1} \oplus C[b_0+1]|_{-(n-(b_0+1))}$. The rest of taps are replaced similarly. Suppose the relationship in \eqref{eq:r2} holds for clock $t = k$, similar to the proof for \eqref{eq:is}, the internal states of the transformed NLFSR at clock $t = k$ can be calculated as

\begin{equation}\label{eq:is2}
    \begin{aligned}
    \hat{x}_{b_0+1}^k \oplus \hat{C}[b_0]|_{-(n-(b_0+1))}^k &= x_{b_0+1}^k \oplus C[b_0]|_{-(n-(b_0+1))}^k \oplus C[b_0]|_{-(n-(b_0+1))}^k \\
    &= x_{b_0+1}^k\\
    \hat{x}_{b_0+2}^k \oplus \hat{C}'[b_0]|_{-(n-(b_0+2))}^k &= \hat{x}_{b_0+2}^k \oplus C[b_0]|_{-(n-(b_0+2))}^k (\hat{x}_1^k, \ldots, \hat{x}_{b_0+1}^k \oplus \hat{C}[b_0]|_{-(n-(b_0+1))}^k)\\
    &= \hat{x}_{b_0+2}^k \oplus C[b_0]|_{-(n-(b_0+2))}^k (x_1^k, \ldots, x_{b_0+1}^k)\\
    &= x_{b_0+2}^k,
    \end{aligned}
\end{equation}
the internal states of rest bits $\hat{x}_i^k \oplus \hat{C}[i]|_{-(n-i)}^k = x_i^k, i \in [b_0+3, n-1]$ are calculated similarly. Therefore, we get 
\begin{equation*}
    \begin{aligned}
    g'_{n-1}(\hat{x}_0^k, \ldots, \hat{x}_{n-1}^k) &= g_{n-1} \oplus m_0 \oplus \ldots \oplus m_{r-1}(\hat{x}_0^k, \ldots, \hat{x}_{b_0}^k, \hat{x}_{b_0+1}^k \oplus \hat{C}[b_0+1]|_{-(n-(b_0+1))}^k,\\
    &\quad \hat{x}_{b_0+2}^k \oplus \hat{C}'[b_0+2]|_{-(n-(b_0+2))}^k \ldots, \hat{x}_{n-1}^k \oplus \hat{C}'[n-1]|_{-(n-(n-1))}^k)\\
    &= g_{n-1} \oplus m_0 \oplus \ldots \oplus m_{r-1}(x_0^k, \ldots, x_{n-1}^k)\\
    g'_{b_j}(\hat{x}_0^k, \ldots, \hat{x}_{b_j}^k) &= m_j|_{-(n-1-b_j)}(\hat{x}_0^k, \ldots, \hat{x}_{b_0}^k, \hat{x}_{b_0+1}^k \oplus \hat{C}[b_0+1]|_{-(n-(b_0+1))}^k,\\
    &\quad \hat{x}_{b_0+2}^k \oplus \hat{C}'[b_0+2]|_{-(n-(b_0+2))}^k \ldots, \hat{x}_{b_j}^k \oplus \hat{C}'[b_j]|_{-(n-b_j)}^k)\\
    &= m_j|_{-(n-1-b_j)}(x_0^k, \ldots, x_{b_j}^k).
\end{aligned}
\end{equation*}
Therefore, we have the internal state $\hat{x}_{n-1}^{k+1}$ of the transformed NLFSR at clock $t = k + 1$ calculated as
\begin{equation*}
    \begin{aligned}
    \hat{x}_{n-1}^{k+1} &= \hat{x}_0^k \oplus g'_{n-1}(\hat{x}_0^k, \ldots, \hat{x}_{n-1}^k)\\
    &= x_{n-1}^{k+1} \oplus g_{n-1}(x_1^k, \ldots, x_{n-1}^k) \oplus g'_{n-1}(\hat{x}_0^k, \ldots, \hat{x}_{n-1}^k)\\
    &= x_{n-1}^{k+1} \oplus m_0 \oplus \ldots \oplus m_{r-1}(x_0^k, \ldots, x_{n-1}^k)\\
    &= x_{n-1}^{k+1} \oplus m_0|_{-1}^{k+1} \oplus \ldots \oplus m_{r-1}|_{-1}^{k+1}\\
    &= x_{n-1}^{k+1} \oplus C[n-1]^{k+1},
    \end{aligned}
\end{equation*}
the internal states $\hat{x}_{b_j}^{k+1}, j \in [0, r]$ are
\begin{equation*}
    \begin{aligned}
    \hat{x}_{b_j}^{k+1} &= \hat{x}_{b_j+1}^k \oplus g'_{b_j}(\hat{x}_0^k, \ldots, \hat{x}_{b_j}^k)\\
    &= x_{b_j+1}^k \oplus C[b_j+1]^k \oplus m_j|_{-(n-1-b_j)}(x_0^k, \ldots, x_{b_j}^k)\\
    &= x_{b_j}^{k+1} \oplus C[b_j]^{k+1},
    \end{aligned}
\end{equation*}
the internal states $\hat{x}_i^{k+1}, i \in [b_0+1, n-2] \textbackslash [b_0, \ldots, b_r]$ are
\begin{equation*}
    \begin{aligned}
    \hat{x}_i^{k+1} &= \hat{x}_{i+1}^k \oplus g'_i(\hat{x}_0^k, \ldots, \hat{x}_i^k)\\
    &= \hat{x}_{i+1}^k \oplus 0\\
    &= x_i^{k+1} \oplus C[i]^{k+1}
    \end{aligned}
\end{equation*}
and the internal states $\hat{x}_i^{k+1}, i \in [0, b_0-1]$ are
\begin{equation*}
    \begin{aligned}
    \hat{x}_i^{k+1} &= \hat{x}_{i+1}^k\\
    &= x_i^{k+1}.
    \end{aligned}
\end{equation*}
Therefore, the relationship in \eqref{eq:r2} holds for internal states at clock $t = k+1$. Since the initial value $X^0$ also satisfy this relationship, we conclude that \eqref{eq:r2} holds for every clock. 

As for the output function, in the algorithm we compensate it iteratively by using $C$ as well. Therefore, the output function of the transformed NLFSR $f'_z$ is the compensation result of the original output function $f_z$, where the tap $x_{b_0+1}$ is replaced by $x_{b_0+1} \oplus C[b_0+1]|_{-(n-(b_0+1))}$. The tap $x_{b_0+2}$ is replaced by $x_{b_0+2} \oplus C[b_0+2]'|_{-(n-(b_0+2))}$, where $C[b_0+2]'|_{-(n-(b_0+2))}$ is $C[b_0+2]|_{-(n-(b_0+2))}$ with $x_{b_0+1}$ replaced by $x_{b_0+1} \oplus C[b_0+1]|_{-(n-(b_0+1))}$. The rest of taps are replaced similarly. According to \eqref{eq:is2}, we have
\begin{equation}
    \begin{aligned}
    f'_z(\hat{x}_0^t, \ldots, \hat{x}_{n-1}^t) =& f_z(\hat{x}_0^t, \ldots, \hat{x}_{b_0}^t, \hat{x}_{b_0+1}^t \oplus \hat{C}[b_0+1]|_{-(n-(b_0+1))}^t,\\
    &\hat{x}_{b_0+2}^t \oplus \hat{C}'[b_0+2]|_{-(n-(b_0+2))}^t \ldots, \hat{x}_{n-1}^t \oplus \hat{C}'[n-1]|_{-(n-(n-1))}^t)\\
    =& f_z(x_0^t, \ldots, x_{n-1}^t).
    \end{aligned}
\end{equation}
Therefore, the output sequences of the two NLFSRs before and after transformation are the same.
\end{proof}

Based on the result of Theorem \ref{Fib}, we are able to reverse the compensation process and transform the Galois NLFSR defined in \eqref{eq:ff} back to the original Fibonacci NLFSR. Formally, we define the feedback function as
\begin{equation} \label{eq:ff2}
    \begin{aligned}
    &f_{n-1} = x_0 \oplus g_{n-1}(x_1, \ldots, x_{n-1})\\
    &f_i = x_{i+1} \oplus g_i(x_0, \ldots, x_i) \quad for \quad i \in [0, n-2]
    \end{aligned}
\end{equation}
and propose a Galois-to-Fibonacci transformation algorithm to convert this kind of Galois NLFSRs to Fibonacci NLFSRs in the following theorem. 

\begin{theorem}\label{Gal}
Given an $n$-bit Galois NLFSR defined in \eqref{eq:ff2} with an output function $f_z(x_0, \ldots, x_{n-1})$ and an initial value $X^0 = \{x_0^0, \ldots, x_{n-1}^0\}$, we shift all the monomials in $g_i$ from $f_i, i \in [0, n-2]$ to $f_{n-1}$. The transformed NLFSR is obtained by following steps:

Step 1: Let the combined compensation list $C = [0, \ldots, 0]$;

Step 2: For each $g_i$, starts from $i = 0$, we remove it from $f_i$ and construct a compensation list $C_i = [0, \ldots, 0, g_i, g_i|_{+1}, \ldots, g_i|_{+(n-i)}]$ and compute $C = C \oplus C_i$;

Step 3: Compensation: we only use $C[i+1]$ to compensate the tap $x_{i+1}$ in all the feedback functions and the output function, which means $x_{i+1}$ is replaced by $x_{i+1} \oplus C[i+1]$;

Step 4: If $i \neq n-2$, then set $i = i + 1$ and go back to Step 2, otherwise, go to Step 5;

Step 5: Now we obtain the final combined compensation list $C$. We xor all the shifted monomials $g_i|_{+(n-i)}, i \in [0, n-2]$ to the feedback function of bit $x_{n-1}$ to get the final feedback functions for the transformed NLFSR. The output function $f'_z$ is constructed by replacing $x_i,i \in [0, n-1]$ in $f_z$ by $x_i \oplus C[i]$, no iteration needed in this compensation process; 

Step 6: The initial value $\hat{X}^0 = \{\hat{x}_0^0, \ldots, \hat{x}_{n-1}\}$ is computed by compensating $X^0 = \{x_0^0, \ldots, x_{n-1}^0\}$ by $C^0$ iteratively starting from $i = 0$ to $i = n-1$.

The transformed NLFSR is a Fibonacci NLFSR and it generate the same sequence as the Galois NLFSR outputs. 
\end{theorem}

\begin{proof}
This theorem is completely the reverse process of the Theorem \ref{Fib}. The proof is omitted. 
\end{proof}

Comparing the feedback functions of the Galois NLFSR in Theorem \ref{Fib} and Theorem \ref{Gal} with the "uniform" Galois NLFSR in \cite{Dubrova09}, it is easy to find that the Galois NLFSR in proposed algorithms is more generalized than the "uniform" Galois NLFSR. In fact, the "uniform" case is included in the \eqref{eq:ff2}. Therefore, the two proposed algorithms both cover the "uniform" case. In this paper, we aim to break the class of Espresso-like ciphers, in which the Galois NLFSR is "uniform". In order to present clearly how the proposed algorithm can cover the "uniform" case, we customize an Uniform\_ Galois-to-Fibonacci algorithm in Theorem \ref{UniGal}.

\begin{theorem}\label{UniGal}
Given an $n$-bit uniform Galois NLFSR defined in Definition \ref{def:uniform} with an output function $f_z(x_0, \ldots, x_{n-1})$ and an initial value $X^0 = \{x_0^0, \ldots, x_{n-1}^0\}$, we shift all the monomials in $g_i$ from $f_i, i \in [\tau, n-2]$ to $f_{n-1}$. The transformed NLFSR is obtained by following steps:

Step 1: For each $g_i, i \in [\tau, n-2]$, we shift it to $f_{n-1}$ and construct a compensation list $C_i = [0, \ldots, 0, g_i, g_i|_{+1}, \ldots, g_i|_{+(n-i)}]$. Then we xor all the $C_i$ as a combined compensation list $C = C_{\tau} \oplus C_{\tau + 1} \oplus \ldots \oplus C_{n-2}$;

Step 2: We use the combined list $C$ to compensate the output function by replacing $x_i$ by $x_i \oplus C[i]$, no iteration needed in this compensation process;

Step 3: The initial value $\hat{X}^0 = \{\hat{x}_0^0, \ldots, \hat{x}_{n-1}\}$ is computed by compensating $X^0 = \{x_0^0, \ldots, x_{n-1}\}$ iteratively by $C^0$ starting from $i = 0$ to $i = n-1$.

The transformed NLFSR is a Fibonacci NLFSR and it generate the same sequence as the Galois NLFSR outputs.
\end{theorem}

\begin{proof}
This algorithm is extracted from Theorem \ref{Gal}. The iterative compensation steps 2 to 4 in Theorem \ref{Gal} are deleted in this algorithm. As we can see in Step 2 and Step 3 in Theorem \ref{Gal}, for each $g_i, i\in [\tau, n-2]$ the compensation is carried out only for tap $x_{i+1}$. However, in the feedback functions of uniform Galois NLFSR, the indexes of any taps in $g_i, i\in [\tau, n-1]$ are equal or smaller than $\tau$. Hence, no compensation is executed.
\end{proof}

\section{Apply Transformation Algorithm}
In this section, we apply the proposed transformation algorithm in Theorem \ref{UniGal} on Galois NLFSRs constructed by Dubrova's scalable method. The result shows that this kind of Galois NLFSRs can always be transformed back into the original Fibonacci LFSR where they are transformed from. For instance, we transform the Galois NLFSR in the Espresso cipher back into a LFSR with a nonlinear output function. 

\subsection{Transform Galois NLFSRs with Period $2^n - 1$}
In Section 2.2, we briefly  introduced how to construct a Galois NLFSR with period $2^n-1$ by using Dubrova's scalable method. The constructed Galois NLFSR is a uniform NLFSR with feedback functions
\begin{equation*}
    \begin{aligned}
    &f_{n-1} = x_0 \oplus g_{n-1}(x_1, \ldots, x_{n-\tau})\\
    &f_{n-2} = x_{n-1} \oplus g_{n-2}(x_0, \ldots, x_{n -\tau - 1})\\
    &\qquad\ldots\\
    &f_{\tau} = x_{\tau + 1} \oplus g_{\tau}(x_0, x_1)\\
    &f_{\tau-1} = x_{\tau}\\
    &\qquad\ldots\\
    &f_0 = x_1,
    \end{aligned}
\end{equation*}

\noindent where $n- \tau \leq \tau$ and an output function $f_z$. 

In this subsection, we apply the Uniform\_Galois-to-Fibonacci transformation algorithm in Theorem \ref{UniGal} on this Galois NLFSR. The transformation process has 3 steps. We first shift all monomials in $g_i, i \in [\tau, n-2]$ to $f_{n-1}$. The feedback functions are transformed to
\begin{equation*}
    \begin{aligned}
    &\hat{f}_{n-1} = x_0 \oplus g_{n-1}(x_1, \ldots, x_{n-\tau})  \oplus g_{\tau}|_{+n-\tau -1} \oplus \ldots \oplus g_{n-2}|_{+1}\\
    &\hat{f}_{n-2} = x_{n-1}\\
    &\qquad\ldots\\
    &\hat{f}_{\tau} = x_{\tau + 1}\\
    &\qquad\ldots\\
    &\hat{f}_0 = x_1.
    \end{aligned}
\end{equation*}
For each $g_i$, we construct a compensation list $C_i = [0, \ldots, 0, g_i, g_i|_{+1}, \ldots, g_i|_{+(n-i)}]$. For example, for $g_\tau$, the compensation list is $C_\tau = [0, \ldots, 0, g_\tau, g_\tau|_{+1}, \ldots, g_\tau|_{+(n-\tau)}]$. Then we get the combined compensation list as 
\begin{equation*}
    \begin{aligned}
    C =& C_\tau \oplus C_{\tau+1} \oplus \ldots \oplus C_{n-2}\\
    =& [0, \ldots, 0, g_\tau, g_\tau|_{+1} \oplus g_{\tau+1},\ldots, \\
    & g_\tau|_{+n-\tau-2} \oplus g_{\tau+1}|_{+n-\tau-3} \oplus \ldots \oplus g_{n-3}|_{+1} \oplus g_{n-2}].
    \end{aligned}
\end{equation*}

In Step 2, we use $C$ to compensate the output function. The resulted output functio is
\begin{equation*}
    \begin{aligned}
    \hat{f}_z =& f_z(x_0, \ldots, x_{\tau}, x_{\tau+1} \oplus C[\tau+1], \ldots, x_{n-1} \oplus C[n-1]\\
    =& f_z(x_0, \ldots, x_{\tau}, x_{\tau+1} \oplus g_\tau, \ldots,\\
    & x_{n-1} \oplus g_\tau|_{+n-\tau-2} \oplus g_{\tau+1}|_{+n-\tau-3} \oplus \ldots \oplus g_{n-3}|_{+1} \oplus g_{n-2}) 
    \end{aligned}
\end{equation*}

The initial value for the transformed NLFSR is calculated according to Step 3 in Theorem \ref{UniGal}.

As we can see, the resulted NLFSR is in Fibonacci configuration. Moreover, since the monomials in $g_i, i \in [\tau, n-2]$ are all shifted from $f_L$ and two $f_N$ when the Galois NLFSR is constructed, the feedback function $\hat{f}_{n-1}$ is actually equal to $f_{n-1} = x_0 \oplus f_L \oplus f_N \oplus f_N$ which is $f_{n-1} = x_0 \oplus f_L$. Therefore, the Galois NLFSR is transformed back into the maximum-length LFSR in Fibonacci configuration with a nonlinear feedback function.

\subsection{Transform Galois NLFSR in Espresso}
As described in Section 2.3, the NLFSR used in Espresso cipher is a 256-bit Galois NLFSR with a 20-variable feedback function. In this section, we apply the Uniform\_Galois-to-Fibonacci transformation algorithm to transform this Galois NLFSR. The transformation process is similar to the application in Section 4.1. Here, we show the details of each step. 

Step 1: We shift all the monomials $\{x_{12}x_{121}, x_{29}x_{52}x_{72}x_{99}, x_8x_{103}, x_5x_{80}, x_6x_{64}, x_4x_{45},\\
x_3x_{32}, x_{50}x_{159} \oplus x_{189}, x_{67}x_{90}x_{110}x_{137}, x_{46}x_{141} \oplus x_{117}, x_{43}x_{118} \oplus x_{103}, x_{44}x_{102} \oplus x_{40}, x_{42}x_{83} \oplus x_8\}$ from $\{f_{193}, f_{197}, f_{201}, f_{205}, f_{209}, f_{213}, f_{217}, f_{231}, f_{235}, f_{239}, f_{243}, f_{247}, f_{251}\}$ back to $f_{255}$ respectively. The feedback functions are transformed to

\begin{equation} \label{eq:ff_Espresso}
    \begin{aligned}
    \hat{f}_{255} &= x_{41}x_{70} \oplus x_{74}x_{183} \oplus x_{87}x_{110}x_{130}x_{157} \oplus \ldots \oplus x_{46}x_{87} \oplus x_{12}\\
    &= x_0 \oplus x_{12} \oplus x_{48} \oplus x_{115} \oplus x_{133} \oplus x_{213}
    \end{aligned}
\end{equation}
and $\hat{f}_i = x_{i+1}$ for the rest of $i \in [0, n-2]$.

The Compensation lists for $g_i, i = \{193, 197, 201, \ldots, 251\}$ are
\begin{equation*}
\begin{aligned}
&C_{193} = [0, \ldots, 0, x_{12}x_{121}, \ldots, x_{73}x_{182}]\\
&C_{197} = [0, \ldots, 0, x_{29}x_{52}x_{72}x_{99}, \ldots, x_{86}x_{109}x_{129}x_{156}]\\
&C_{201} = [0, \ldots, 0, x_8x_{103}, \ldots, x_{61}x_{156}]\\
&C_{205} = [0, \ldots, 0, x_5x_{80}, \ldots, x_{54}x_{129}]\\
&C_{209} = [0, \ldots, 0, x_6x_{64}, \ldots, x_{51}x_{109}]\\
&C_{213} = [0, \ldots, 0, x_4x_{45}, \ldots, x_{45}x_{86}]\\
&C_{217} = [0, \ldots, 0, x_3x_{32}, \ldots, x_{40}x_{69}]\\
&C_{231} = [0, \ldots, 0, x_{50}x_{159} \oplus x_{189}, \ldots, x_{73}x_{182} \oplus x_{212}]\\
&C_{235} = [0, \ldots, 0, x_{67}x_{90}x_{110}x_{137}, \ldots, x_{86}x_{109}x_{129}x_{156}]\\
&C_{239} = [0, \ldots, 0, x_{46}x_{141} \oplus x_{117}, \ldots, x_{61}x_{156} \oplus x_{132}]\\
&C_{243} = [0, \ldots, 0, x_{43}x_{118} \oplus x_{103}, \ldots, x_{54}x_{129} \oplus x_{114}]\\
&C_{247} = [0, \ldots, 0, x_{44}x_{102} \oplus x_{40}, \ldots, x_{51}x_{109} \oplus x_{47}]\\
&C_{251} = [0, \ldots, 0, x_{42}x_{83} \oplus x_8, \ldots, x_{45}x_{86} \oplus x_{11}].
\end{aligned}
\end{equation*}

The combined compensation list is
\begin{equation*}
    \begin{aligned}
    C &= [0, \ldots, 0, x_{12}x_{121}, \ldots, x_{16}x_{125} \oplus x_{29}x_{52}x_{72}x_{99}, \ldots,\\
    & \quad x_{40}x_{69} \oplus x_{212} \oplus x_{132} \oplus x_{114} \oplus x_{47} \oplus x_{11}].
    \end{aligned}
\end{equation*}

Step 2: We use $C$ to compensate the output function. The output function becomes
\begin{equation*}
    \begin{aligned}
    \hat{f}_z =& x_{80} \oplus x_{99} \oplus x_{137} \oplus \hat{x}_{227} \oplus \hat{x}_{222} \oplus x_{187} \oplus \hat{x}_{243}\hat{x}_{217} \oplus \hat{x}_{247}\hat{x}_{231} \oplus \hat{x}_{213}\hat{x}_{235} \\
    &\oplus \hat{x}_{255}\hat{x}_{251} \oplus x_{181}\hat{x}_{239} \oplus x_{174}x_{44} \oplus x_{164}x_{29} \oplus \hat{x}_{255}\hat{x}_{247}\hat{x}_{243}\hat{x}_{213}x_{181}x_{174},
    \end{aligned}
\end{equation*}
where $\hat{x}_i$ for $i = \{213, 217, 222, 227, 231, 235, 239, 243, 247, 251, 255\}$ denotes $x_i$ replaced by $x_i \oplus C[i]$. Specifically, 
\begin{equation*}
    \begin{aligned}
    \hat{x}_{213} &= x_{213} \oplus x_{31}x_{140} \oplus x_{44}x_{67}x_{87}x_{114} \oplus x_{19}x_{114} \oplus x_{12}x_{87} \oplus x_9x_{67}\\
    \hat{x}_{217} &= x_{217} \oplus x_{35}x_{144} \oplus x_{48}x_{71}x_{91}x_{118} \oplus x_{23}x_{118} \oplus x_{16}x_{91} \oplus x_{13}x_{71} \oplus x_7x_{48}\\
    \hat{x}_{222} &= x_{222} \oplus x_{40}x_{149} \oplus x_{53}x_{76}x_{96}x_{123} \oplus x_{28}x_{123} \oplus x_{21}x_{96} \oplus x_{18}x_{76} \oplus x_{12}x_{53} \oplus x_7x_{36}\\
    \hat{x}_{227} &= x_{227} \oplus x_{45}x_{154} \oplus x_{58}x_{81}x_{101}x_{128} \oplus x_{33}x_{128} \oplus x_{26}x_{101} \oplus x_{23}x_{81} \oplus x_{17}x_{58} \oplus x_{12}x_{41}\\
    \hat{x}_{231} &= x_{231} \oplus x_{49}x_{158} \oplus x_{62}x_{85}x_{105}x_{132} \oplus x_{37}x_{132} \oplus x_{30}x_{105} \oplus x_{27}x_{85} \oplus x_{21}x_{62} \oplus x_{16}x_{45}\\
    \hat{x}_{235} &= x_{235} \oplus x_{66}x_{89}x_{109}x_{136} \oplus x_{41}x_{136} \oplus x_{34}x_{109} \oplus x_{31}x_{89} \oplus x_{25}x_{66} \oplus x_{20}x_{49} \oplus x_{192}\\
    \hat{x}_{239} &= x_{239} \oplus x_{45}x_{140} \oplus x_{38}x_{113} \oplus x_{35}x_{93} \oplus x_{29}x_{70} \oplus x_{24}x_{53} \oplus x_{196}\\
    \hat{x}_{243} &= x_{243} \oplus x_{42}x_{117} \oplus x_{39}x_{97} \oplus x_{33}x_{74} \oplus x_{28}x_{57} \oplus x_{200} \oplus x_{120}\\
    \hat{x}_{247} &= x_{247} \oplus x_{43}x_{101} \oplus x_{37}x_{78} \oplus x_{32}x_{61} \oplus x_{204} \oplus x_{124} \oplus x_{106}\\
    \hat{x}_{251} &= x_{251} \oplus x_{41}x_{82} \oplus x_{36}x_{65} \oplus x_{208} \oplus x_{128} \oplus x_{110} \oplus x_{43}\\
    \hat{x}_{255} &= x_{255} \oplus x_{40}x_{69} \oplus x_{212} \oplus x_{132} \oplus x_{114} \oplus x_{47} \oplus x_{11}.
    \end{aligned}
\end{equation*}

Step 3: The initial value $\hat{X}^0$ for the transformed NLFSR is calculated by compensating $X^0 = \{x_0^0, \ldots, x_{255}^0\}$ by $C^0$ iteratively starting from $i = 0$ to $i = 255$.

As we can see in the final feedback functions \eqref{eq:ff_Espresso}, the feedback only fed to the last stage $x_{255}$ and the feedback function $f_{255}$ only contains linear terms. Therefore, the transformed NLFSR is a Fibonacci LFSR. The Espresso cipher is actually equivalent to a LFSR filter generator. Specifically, the nonlinear output function consists of 2289 monomials. There are 104 variables in the function and the algebraic degree is 12. With the corresponding initial value, transformed feedback functions and transformed output function, the transformed LFSR generates the same output sequence as the Galois NLFSR G in Espresso cipher does. 

\section{Cryptanalysis}
In this section, we conduct cryptabalysis on transformed version of the Espresso stream cipher, which is a LFSR filter generator according to the result in last section. The security analysis of LFSR filter generators has been thoroughly studied in literature. Several powerful cryptanalytic attacks such as algebraic attacks and correlation attacks have been proposed to break LFSR filter generators. 

\subsection{Algebraic Attack}
Algebraic attack is a very powerful cryptanalysis technique to break LFSR filter generators. The basic idea is to build a system of equations connecting the keystream bits and the initial state of the LFSR, and then solve these equations to recover the secret key in the initial state. The equations are obtained either directly using the output function or multiplying it with a well chosen multivariate equation to lower the degree of the output function. 

\subsubsection{Standard Algebraic Attack.}
The most representative algebraic attack is proposed by Courtois~\cite{CourtoisM03,Courtois03}. A detailed analysis of complexity of this standard algebraic attack is re-estimated by Hawkes and Rose~\cite{HawkesR04}. We apply this attack on the transformed LFSR in Section 4. The attack process and the complexity of each step are presented below.

Step 1: Form a system of equations. As presented in Section 4, the output function of the transformed LFSR is a 12 degree Boolean function with 104 variables. First, we consider reducing the algebraic degree of the equations which we get from the output function to relate the output keystream bits to the initial state of the LFSR. From the function $\hat{f}_z$, we observe that monomials with highest degree are included in the term $\hat{x}_{255}\hat{x}_{247}\hat{x}_{243}\hat{x}_{213}x_{181}x_{174}$ where $\hat{x}_{255}$, $\hat{x}_{247}$, $\hat{x}_{243}$ and $\hat{x}_{213}$ are represented in previous section. This term can be expended into 2058 monomials. It is noteworthy that the bit $x_{181}$ and $x_{174}$ appear in all these monomials. Therefore, we can multiply the output function with $(x_{181}+1)$ or $(x_{174}+1)$ and obtain a Boolean function with lower degree. Moreover, we can further reduce the degree by multiplying it with $(x_{44}+1)$ or $(x_{67}+1)$ or $x_{87}+1)$ or $(x_{114}+1)$ or $(x_{66}+1)$ or $(x_{89}+1)$ or $(x_{109}+1)$ or $x_{136}+1)$. In total we obtain 16 functions with degree 8. One of such function is sufficient to form a system of multivariate equations we need for the attack. For instance, we multiply the output function $f_z$ with $g = (x_{181}+1)(x_{44}+1)$ to get function $h$

\begin{center}
    $h = f_zg = f_z(x_{181}+1)(x_{44}+1)$.
\end{center}

The degree of the output function $f_z$ is denoted as $d_f = 12$, the degree of $g$ is $e = 2$ and the degree of $h$ is $d = 8$. For each known keystream bit $z^t$ at clock $t$, we derive equation

\begin{center}
    $h(\vec X^t) = z^tg(\vec X^t)$.
\end{center}

Since the keystream $z^t$ is a binary bit and only affect $g(X^t)$, we rewrite the equation as

\begin{equation}\label{eq:12}
\begin{aligned}
    h(\vec X^t) + g(\vec X^t, z^t) = 0.
\end{aligned}
\end{equation}
The complexity for this step can be ignored. 

Step 2: Pre-computation. As described in~\cite{HawkesR04}, any Boolean function of the LFSR state can be expressed as a product of the monomial state matrix $\vec M_d$ with a row vector of that Boolean function such as $\vec h$ and $\vec g$. Mapping from one monomial state to the next monomial state can be expressed as a matrix product $\vec M_d^{t+1} = \vec R_d\vec M_d^t$. Moreover, for every clock $t$, the monomial state of the LFSR can be expressed as $\vec M_d^t = \vec R_d^t \cdot \vec M_d^0$. We consider the Boolean functions $h$ and $g$ as depending on distinct monomial states $\vec M_d$ and $\vec M_e$, with corresponding monomial state rewriting matrices $\vec R_d$ and $\vec R_e$~\cite{HawkesR04}. $\vec M_d$ represents all $D = \sum_{i = 0}^d {n\choose i}$ monomials of degree $d$ or less and $M_e$ represents all $E = \sum_{i = 0}^e {n\choose i}$ monomials of degree $e$ or less.  Equation \eqref{eq:12} is then rewritten as

\begin{equation} \label{eq:13}
    h(\vec R_d^t\cdot \vec M_d^0) + g(\vec R_e^t\cdot \vec M_e^0) = 0 \vec h\cdot \vec R_d^t\cdot \vec M_d^0 + \vec g(\vec z^t) \cdot \vec R_e^t\cdot \vec M_e^0 = 0.
\end{equation}

For the monomial state rewriting matrix $\vec R_d$, there exists a characteristic polynomial $p(x) = \sum_{i=0}^{D}p_ix_i$ so that

\begin{center}
    $\sum_{i=0}^{D}p_i \cdot \vec R_d^i = \vec 0$,
\end{center}
where $\vec0$ represents the all-zero matrix. The characteristic polynomial can be computed by the Theorem 1 and the steps presented in~\cite{HawkesR04} inspired by~\cite{Key76}.

We find the linear combination such that

\begin{center}
    $\sum_{i=0}^{D}p_i \cdot (\vec h \cdot \vec R_d^{t+i}\cdot \vec M_d^0+ \vec g(\vec z^{t+i}) \cdot \vec R_e^{t+i} \cdot \vec M_e^0) = 0$.
\end{center}

The equation can be further expanded as

\begin{center}
    $ \vec h \cdot \vec R_d^t \cdot \vec M_d^0 \cdot \sum_{i=0}^{D}p_i \cdot \vec R_d^i + \vec M_e^0 \cdot \sum_{i=0}^{D}p_i \cdot \vec g(\vec z^{t+i}) \cdot \vec R_e^{t+i} = 0$.
\end{center}

From \eqref{eq:13}, we have

\begin{center}
    $\vec M_e^0 \cdot \sum_{i=0}^{D}p_i \cdot \vec g(\vec z^{t+i}) \cdot \vec R_e^{t+i} = 0$.
\end{center}

The complexity of this step depends on the computation of the characteristic polynomial, which requires $c \cdot D(n(log n)^2 + (log_2D)^3)$ operations for small constant $c$.

Step 3: Substitution. The equation in step 2 is performed component-wise, so we write the equation for $t = 1, \ldots, E$. Then the Fast Fourier Transform (FFT) is applied to substitute the obtained keystream bits into the equations. The complexity of this step is $2DE(log_2E)$.

Step 4: Solving the equations. The last step is to solve the equations by linearization. The complexity of this step is $E^\omega$, where $\omega = 2.807$ is the exponent of the Gaussian reduction~\cite{RonjomH07}. 

Among all the 4 steps, the dominant complexity is in the Step 3. The complexity of this attack is $2ED(log_2E)$. The number of required keystream bits is $D+E-1$. Therefore, we can break Espresso cipher by using the standard algebraic attack with $2^{48.58}$ keystream bits and the computation complexity is $\mathcal{O}(2^{68.44})$.

\subsubsection{R\o njom-Helleseth Attack.}
We also consider the attack proposed by R\o njom-Helleseth~\cite{RonjomH07}. In this attack, the idea of expressing the output function in terms of monomials of initial state and finding the characteristic polynomial of the coefficient sequences of all monomials is also adapted. The main difference from the standard algebraic attack described in Section 5.1 is in the step of   solving the equations. This attack only needs to solve the linear part of the formed system of equations. The process of this attack includes:

Step 1: Pre-computation. From the output function, compute the characteristic polynomial $p(x) = \sum_{j=0}^{D-n}p_jx^j$ of the coefficient sequences of all monomial of degree $\geq 2$. The complexity of this step is $\mathcal{O}(D(log_2D)^3)$.

Step 2: Form a linear equation system. For each output keystream bit, we represent it as $z_t = f_t(x_0^0, \ldots, x_{n-1}^0)$. From the characteristic polynomial $p(x)$, we form a system of equations

\begin{center}
    $z^*_t = f^*_t(x_0^0, \ldots, x_{n-1}^0)$ for $t = 0, 1, \ldots, n-1$,
\end{center}
where 

\begin{center}
    $z^*_t = \sum_{j=0}^{D-n}p_jz_{t+j}$
\end{center}
and 
\begin{center}
    $f^*_t = \sum_{j=0}^{D-n}p_jf_{t+j}(x_0^0, \ldots, x_{n-1}^0)$.
\end{center}

The complexity of this step is determined by the calculation of $z^*_t$ and $f^*_t$, which is $\mathcal{O}(D)$.

Step 3: Solving the system of linear equations. The complexity of this step is $\mathcal{O}(n^\omega)$.

In conclusion, by using this attack, we break the Espresso cipher with computation complexity of $\mathcal{O}(2^{66.86})$, pre-computation complexity $\mathcal{O}(2^{84.97})$ and $2^{66.86}$ keystream bits. In order to resist the algebraic attack, the degree of the output function in the transformed version of the cipher should be at least 30 and the output function should have large algebraic immunity. 

\subsection{Overall Security Analysis}
Another powerful attack against LFSR filter generators is correlation attack. The basic idea is to derive linear approximation equations of the nonlinear output function to recover the initial state~\cite{BerbainGM06,Maximov06}. In 2010, R\o njom and Cid \cite{RonjomC10} investigated the nonlinear equivalence of LFSR-based stream ciphers. In the paper, they showed how to change the primitive root of the underlying finite field to obtain an equivalent filter generators. They also pointed out that current analysis of distance from a nonlinear function to the space of affine functions is incomplete with respect to LFSR-based stream ciphers. They redefine the nonlinearity of the LFSR filter function as the minimal distance between equivalent filter functions and affine functions, which implies that a correlation attack may be more successful on a weak equivalent generator. This observation has been proved by Canteaut and Rotella at FSE 2016 \cite{CanteautR16}. In their paper, they describe how to mount a fast correlation attack on equivalent LFSR filters. By following their result and Ronjom and Cid’s idea of changing the primitive root of the underlying finite field to obtain an equivalent filter generator, a fast correlation attack is a potential threat to the Espresso cipher. 

In the design of Espresso cipher, the security analysis is conducted on a NLFSR F instead of the original Galois NLFSR G. As the author claimed~\cite{DubrovaH17}, F is equivalent to the original Galois NLFSR G. However, according to our proposed transformation algorithm, the two NLFSRs F and G are equivalent only if the output function of F is changed accordingly. Therefore, the security analysis of the cipher is not conducted on the actual cipher. Whether it is resistant to chosen IV attacks, differential attacks and weak key attacks needs to be reassessed. Moreover, a known attack against Espresso cipher proposed in~\cite{WangL17} is a related key chosen IV attack. This attack is also not mounted on the original Galois NLFSR G but on another NLFSR based cipher denoted as Espresso-a. Nevertheless, the two ciphers would not generate same sequence when the output functions are the same. The output function of Espresso-a should be transformed by the proposed algorithm in Theorem \ref{UniGal}. Whether the Espresso cipher is weak against this attack or not should be reevaluated. 

The idea of using Galois NLFSRs constructed from maximum length LFSR to build stream ciphers can certainly optimise the hardware size and the throughput of the cipher. However, this innovative design methodology exposes weaknesses to existing cryptanalytic attacks. Therefore, this kind of design method should be avoided in the future. 

\section{Conclusion}
In this paper, we analyse the security of Espresso-like stream ciphers based on Galois NLFSRs constructed from maximum length LFSRs. Considering the problem in existing transformation algorithms, we develop the idea of compensating the output function and feedback functions in the cipher and propose a Fibonacci-to-Galois and a Galois-to-Fibonacci transformation algorithm. Then we customize an Uniform\_Galois-to-Fibonacci algorithm to transform the Galois NLFSR in Espresso-like cipher. The result shows that the transformed cipher is actually a linear filter generator based stream cipher. As an example, the Espresso cipher is transformed into a LFSR filter generator with a nonlinear output function. The transformed cipher is broken by the standard algebraic attack and the R\o njom-Helleseth attack. At last, we analyse the security of this kind stream cipher from the overall perspective and point out that this kind of design method should be avoided in the future.

\bibliographystyle{splncs04}
%

\end{document}